\documentclass[
submission
]{dmtcs-episciences}

\usepackage[utf8]{inputenc}
\usepackage{subfigure}
\fussy

\usepackage[round]{natbib}

\usepackage{amssymb,amsmath}
\usepackage{amsthm}
\usepackage{appendix}
\usepackage[margin=1in]{geometry}
\usepackage{url}

\usepackage{subfigure}

\newcommand\camelia[1]{\textcolor{blue}{#1}}

\newtheorem{theorem}{Theorem}
\newtheorem{corollary}[theorem]{Corollary}
\newtheorem{definition}[theorem]{Definition}
\newtheorem{proposition}[theorem]{Proposition}

\date{\vspace*{-2em}}

\usepackage{array}
\usepackage{multirow}
\usepackage{comment}
\usepackage{color}
\usepackage{longtable}
\usepackage[boxruled,linesnumbered]{algorithm2e}
\usepackage{algorithmic}
\usepackage{graphicx}
\graphicspath{ {./images/} } 
\usepackage[rightcaption]{sidecap}
\usepackage{paralist}
\usepackage{natbib}
\newtheorem{observation}{Observation}

\newtheorem{open}{Open question}

\let\oldnl\nl
\newcommand{\nonl}{\renewcommand{\nl}{\let\nl\oldnl}}

\title{Approximate and exact results for the harmonious chromatic number}

\author{Ruxandra Marinescu-Ghemeci\affiliationmark{1} \and Camelia Obreja\affiliationmark{1} \and Alexandru Popa\affiliationmark{1,2}}

\affiliation{
        Department of Computer Science, University of Bucharest \\
        National Institute for Research and Development in Informatics
}


\received{}
\revised{}
\accepted{}

\begin{document}
\publicationdetails{VOL}{2021}{ISS}{NUM}{SUBM}

\maketitle



\begin{abstract}
Graph colorings is a fundamental topic in graph theory that require an assignment of labels (or colors) to vertices or edges subject to various constraints.  We focus on the harmonious coloring of a graph, which is a proper vertex coloring such that for every two distinct colors $i$, $j$ at most one pair of adjacent vertices are colored with $i$ and $j$. This type of coloring is edge-distinguishing and has potential applications in transportation network, computer network, airway network system.

The results presented in this paper fall into two categories: in the first part of the paper we are concerned with the computational aspects of finding a minimum harmonious coloring and in the second part we determine the exact value of the harmonious chromatic number for some particular graphs and classes of graphs.
More precisely, in the first part we show that finding a minimum harmonious coloring for arbitrary graphs is APX-hard, the natural greedy algorithm is a $\Omega(\sqrt{n})$-approximation, and, moreover, we show a relationship between the vertex cover and the harmonious chromatic number. 
In the second part we determine the exact value of the harmonious chromatic number for all  $3$-regular planar graphs of diameter $3$, some non-planar regular graphs and cycle-related graphs.

\end{abstract}

\section{Introduction}

A key topic in the area of graph theory is represented by graph colorings. 
The proper vertex $k$-coloring is perhaps the most famous type of coloring and has many applications such as  scheduling, pattern matching, exam timetabling, seating plans design (see ~\cite{zhang2016kaleidoscopic,utility}). 
There are numerous types of colorings, e.g., harmonious, graceful, metric, sigma, set, multiset (see~\cite{zhang2016kaleidoscopic} and the references therein).

In this paper we focus on harmonious colorings. We consider only finite undirected  graphs $G(V,E)$, with $|V|$ vertices (or nodes) and $|E|$ edges. Given a graph $G$, we denote by $V(G)$ the set of vertices of $G$ and by $E(G)$ the set of edges of $G$, respectively. Given a positive integer $k$, let $[k] = \{1, 2, \dots, k\}$.

\subsection{Preliminaries and previous work}
  
The concept of harmonious coloring was proposed independently by~\cite{frank1982line} and by~\cite{hopcroft1983} and defined below.

\begin{definition}[Harmonious coloring]
~\label{HarmoniousColoring}
Let $G$ be a graph and $c: V(G)\rightarrow [k]$ be a proper vertex coloring of 
$G$. The coloring $c$ is called  harmonious if for every two distinct colors $i, j \in [k]$ there is at most one pair of adjacent vertices in $G$ colored with $i$ and $j$.
\end{definition}

The definition of harmonious coloring leads to next observation.
\begin{observation}
~\label{obs:harmonious-distinct}
There exists at least one harmonious coloring in any graph, since coloring all vertices with distinct colors produces a harmonious coloring.
\end{observation}


The existence of the harmonious coloring of the graph follows from Observation~\ref{obs:harmonious-distinct}. We are interested in finding the minimum number of colors required to have a valid harmonious coloring, that is to find the harmonious chromatic number of a graph, as defined next.

\begin{definition}[The harmonious chromatic number]
~\label{Harmonious_chromatic_number}
The minimum positive integer $k$ for
which a graph $G$ has a harmonious $k$-coloring is called the harmonious chromatic number of $G$ and is denoted by $h(G)$.
\end{definition}

We can associate to a harmonious $k$-coloring $c$ of $G$ an edge coloring $c'$ of $G$ as follows: each edge $uv$ is assigned the color $c'(uv)=\{c(u),c(v)\}$. A color $c'(uv)$  is a $2$-element subset of the set of colors assigned to the vertices of $G$. In the resulting edge coloring $c'$ all the edges are colored with distinct colors. Thus, it follows that ${k \choose 2} \geq |E|$. 

Note that the harmonious coloring is different than the harmonious labeling of a graph, introduced by~\cite{GrahamSloane}. In a harmonious labeling $c$ of an undirected graph $G$ the colors of vertices are elements of $\mathbb Z_k$ (set of integers modulo $k$) and the induced edge-coloring $c'$ is defined as $c'(uv)= (c(u)+c(v)) (\mbox{mod } k) $.

\subsubsection{Known results related to the computational complexity of the harmonious coloring problem
}

\cite{hopcroft1983} show that the harmonious coloring problem for arbitrary graphs is NP-complete. Moreover, determining whether a graph has a harmonious coloring using at most $k$ colors is known to be NP-complete even in trees (\cite{edwards1995complexity}), split graphs (\cite{interval_permutation}), interval graphs (\cite{interval_permutation,bodlander1989cograph}) and several other classes of graphs (\cite{bodlander1989cograph,edwards_1997,edwards1995complexity,interval_permutation,subclasses_2010,asdre_2007_NP_classes}). Polynomial time algorithms are known for some special classes of graphs (\cite{miller1991}), the most important being for trees of bounded degree (\cite{edwards_1996}).

A recent paper that deals with the computational aspects of harmonious coloring is (\cite{kolay2019harmonious}). In this paper the authors list the classes of graphs for which the harmonious coloring is known to be NP-hard. 

\cite{kolay2019harmonious} study the parameterized complexity of the harmonious coloring problem under various parameters such as solution size,  above or below known guaranteed bounds and   vertex cover number of the graph.

\subsubsection{Known upper and lower bounds for the harmonious chromatic number}

By Observation \ref{obs:harmonious-distinct}, we have $h(G)\le |V(G)|$. Lower bounds for the harmonious chromatic number of a graph $G$ of size $m$ and maximum degree $\Delta$ are given in~\cite{zhang2016kaleidoscopic} and stated next.

\begin{theorem}[\cite{zhang2016kaleidoscopic}]\label{lower_m}
If G is a graph of size m, then \begin{equation*} h(G) \geq \left \lceil  \frac  {1+ \sqrt{8m+1}}{2} \right \rceil  .\end{equation*}
\end{theorem}

We recall the following theorem that relates the harmonious chromatic number and the maximum degree of a graph.

\begin{theorem}[\cite{zhang2016kaleidoscopic}] \label{the:delta}
If $G$ is a graph having maximum degree $\Delta$, then \begin{equation*}
    h(G) \geq \Delta + 1 . \end{equation*}
\end{theorem}

\begin{corollary}
\label{cor: complementary}
 For a graph G of order $n \geq 2$, $h(G)= 1$  if and only if $ G =\overline { K_n }$. Furthermore,
$h(K_n) = n$. 
\end{corollary}

\begin{corollary}\label{lower_bond}
Any graph of order $n$ having maximum degree $n-1$ has harmonious chromatic number $n$.
\end{corollary}

\cite{LeeMitchem} present an upper bound for the harmonious chromatic number of a graph.
\begin {theorem}
[\cite{LeeMitchem}]\label{uper_bound} If $G$ is a graph of order n having maximum degree $\Delta$, then
\begin{equation*}
h(G)\leq (\Delta^2+1)\lceil \sqrt n \rceil.\end{equation*}
\end{theorem}

\cite{McDiarmid1991} determined an improved upper bound for the harmonious chromatic number of a
graph.
\begin{theorem}
[\cite{McDiarmid1991}]
If $G$ is a nonempty graph of order $n \geq 2$ having maximum
degree $\Delta$, then
\begin{equation*}
h(G)\leq 2 \Delta \sqrt {n-1} .\end{equation*}
\end{theorem}

\subsubsection{Previous results for harmonious 
chromatic number on particular classes of graphs}
\label{sec:previous_particular}
 
Concerning \emph{the exact value of the harmonious chromatic number of a graph}, there are only few graphs for which the precise value of the harmonious chromatic number is known. 
The harmonious chromatic number of the path with $n$ vertices $P_n$  has been determined by~\cite{lu1}, and of cycles $C_n$ by~\cite{Mitchem1}. 
The harmonious chromatic number of a class of caterpillars with at most one vertex of degree more than $2$ (paths, stars, shooting stars and comets), and an upper bound of the harmonious chromatic number of $3$-regular caterpillars were found by~\cite{article}. 
Harmonious coloring has been studied for distance degree regular graphs of diameter $3$ and for several particular classes of graphs such as Parachute, Jellyfish, Gear, and Helm graph by~\cite{huilgol2016harmonious}. 

The harmonious chromatic number for the central graph, middle graph, and total graph of some families of graphs was studied in various papers: 
prism graph by~\cite{article_MYn_CYn}; flower graph
, belt graph
, rose graph 
 and steering graph 
 by~\cite{central};
snake derived architecture 
by~\cite{snake_graphs}
; Jahangir graph by~\cite{jahangir_graph}; 
star graph by~\cite{line_star}, and 
double star graph by~\cite{double_star}. 

Next we present a couple of known results related to graphs with diameter $2$.
Recall that the \textsl{distance} $d(u,v)$ between two vertices is the length of a shortest $u-v$ path in a graph $G(V,E)$, and the \textsl{diameter} $diam(G)$ is the largest distance between any two vertices of $G$.

\begin{theorem}[folklore]
\label{the:dist2} Let $G$ be a graph with diameter 2 and $v$ be an arbitrary vertex of $G$. Denote by $N_2[v]$ the set of vertices at distance at most $2$ from $v$, including $v$. Then, in a harmonious coloring of $G$, vertices from $N_2[v]$ receive distinct colors.
\end{theorem}

\begin{proof} Let $c$ be a harmonious coloring of $G$. 
Assume by contradiction that two vertices $a,b \in N_2[v]$ have the same color, i.e. $c(a) = c(b)$. Then we have two cases.  In the first case, $a$ and $b$ are adjacent vertices; the fact that they have the same color contradicts the definition of a harmonious coloring. In the second case, assume that $a$ and $b$ are at distance $2$. Thus, there exists a vertex $x \neq a, x \neq b$ such that $(a,x) \in E$ and $(b,x) \in E$. Since $a$ and $b$ have the same color, it follows that  
$\{c(a),c(x)\} = \{c(b),c(x)\}$, which again contradicts the definition of a harmonious coloring. Thus, the theorem holds.
  \end{proof}

\begin{corollary}[folklore]\label{cor:diameter_2} Any graph $G$ with $n$ vertices and diameter $2$ has the harmonious chromatic number $n$. 
\end{corollary}

\begin{proof}
In a graph with diameter $2$, all the vertices are at distance at most $2$ and, thus, according to Theorem~\ref{the:dist2}, all the vertices in the graph must receive distinct colors.
\end{proof}

Among the most known graphs with diameter two are individual graphs like complete bipartite graph $K_{3,3}$, Wagner graph, Moser Spindle graph, Golden-Harary graph, Fritsch graph, Petersen graph, house graph, prism graph $Y_3$, octahedron graph and some classes of graphs like cographs, the friendship graphs, the fan graphs, the wheel graphs.

\subsection{Our results}

In this paper we show the following results. In Section~\ref{sec:computation} we tackle the harmonious coloring problem from the computational point of view. More precisely, we show that the harmonious coloring problem cannot be approximated within a factor of $1.17 - \epsilon$, assuming $P \neq NP$ and within a factor $4/3 - \epsilon$, assuming the Unique Games Conjecture, $\forall \epsilon > 0$. We prove our hardness results by generalizing the NP-hardness reduction of~\cite{complete}. 
We also show why the natural greedy algorithm (that colors vertices one by one and assigns the smallest color possible) is not a good approximation. In the last part of Section~\ref{sec:computation}, we present a relation between the harmonious chromatic number and the minimum vertex cover.

Then, in Section~\ref{sec:particular_classes} we determine exact values of the harmonious 
chromatic number for particular classes of graphs, like $(3,3)$-regular planar and non-planar graphs and some families of cycle-related graphs. Some of these results are obtained using a backtracking based computer program.

\section{Computational results on harmonious coloring}
\label{sec:computation}

In this section we aim to tackle the computational complexity of the harmonious coloring. 

\subsection{Hardness of approximation of harmonious coloring on general graphs}

In this subsection we show that the harmonious coloring APX-hard or that it does not admit a polynomial time approximation scheme. In other words, there exists a constant $c$ such that the harmonious coloring number on general graphs cannot be approximated within a factor of $c$. 

\begin{theorem}
There exists a constant $c < 1.17$ such that the harmonious coloring problem cannot be approximated within a factor of $c$, unless $P = NP$. Moreover, if we assume the Unique Games Conjecture, the harmonious coloring problem cannot be approximated within a factor of $4/3 - \epsilon$ for any $\epsilon > 0$.
\end{theorem}

\begin{proof}
We show our result via a reduction from the Independent Set problem. Our reduction is a simple modification of the reduction of~\cite{complete}. Given a graph $G = (V,E)$ for which we aim to find an independent set with $k\le |V|$ elements, we can construct in polynomial time an instance of the harmonious coloring problem for a graph with two connected components $G'$ and $G''$. The first component $G'$ has vertex set $V \cup \{v_1, v_2, v_3\}$. The set of edges of $E(G')$ is obtained by adding at $E(G)$ edges between every vertex of $G$ and $v_1, v_2,$ and $v_3$, respectively, and edges $\{v_1,v_2\}$, $\{v_2,v_3\}$, $\{v_1,v_3\}$. The second component $G''$ is a clique with $|V|$ vertices. 

Observe that $G'$ cannot be harmoniously colored with less than $|V| + 3$ colors, since it has diameter less or equal than $2$ (Corollary \ref{cor:diameter_2}). 

The claim is that this two-component graph can be harmoniously colored with $2|V| + 3 - k$ colors if and only if $G$ has an independent set of size $k$.

Assume first that $G$ has an independent set $X$  of size $k$. 
We define a harmonious coloring for the two-component graph as follows: color vertices of $G'$ with distinct colors; then color $|X|$ vertices of $G''$ with the colors used for the vertices of $X$ in $G'$ and the rest of the vertices of $G''$ with $|V|-|X|$ new colors. The obtained coloring is obviously harmonious and uses $|V|+3+|V|-|X|=2|V|+3-|X|$ colors.

Conversely, assume that the two-component graph has a harmonious coloring with $2|V| + 3 - k$ colors.   For $k=1$ there is obviously an independent set of size $k$ in $G$. Assume $k\ge 2$. For the vertices in component $G'$ exactly $|V|+3$ distinct colors are used (Corollary \ref{cor:diameter_2}).  We have $|V|-k$ unused colors left only for vertices in $G''$. Since $G''$ is a clique, vertices from $G''$ have distinct colors. It follows that there are $k$ colors used both for vertices in $G'$ and in $G''$. By the definition of a harmonious coloring, it follows that in $G'$ these vertices form an independent set. This independent set is also an independent set in $G$, since vertices $v_1$, $v_2$, $v_3$ are pairwise adjacent and adjacent to all the vertices in $G$.

Let $0 < s < c \le \frac{1}{2}$ be constants and let $GapIS(c,s)$ be a ``promise gap problem'' where an $n$-vertex graph is given with the promise that either it contains an independent set of size $cn$ or contains no independent set of size $sn$ and the algorithmic task is to distinguish between the two cases. 
According to our reduction, we have that if $GapIS(c,s)$ is NP-hard, then the harmonious coloring is NP-hard to approximate within $$ \frac{2|V| + 3 - s|V|}{2|V| + 3 - c|V|}. $$ Thus,  harmonious coloring is NP-hard to approximate within $\frac{2-s}{2-c} + \epsilon$, for some $\epsilon > 0$. 

The best gap known is of~\cite{Dinur04onthe} and has $GapIS(1 - 2^{-1/d} - \epsilon, \epsilon)$ for $d \ge 2$. Thus, for $d=2$, we have that the harmonious coloring is hard to approximate within $\frac{2}{1 + \frac{1}{\sqrt{2}}} \approx 1.17$, unless $P = NP$. Then, according to~\cite{khot2008vertex}, assuming the Unique Games Conjecture we have $GapIS(1/2 - \epsilon, \epsilon)$. Thus, assuming the Unique Games Conjecture, the harmonious coloring problem is hard to approximate within a factor of $4/3 - \epsilon$.
  \end{proof}

\subsection{The natural greedy algorithm is an $\Omega(\sqrt{n})$-approximation }

A natural greedy algorithm to harmoniously color a graph is as follows. Process the vertices arbitrarily and color each vertex with the smallest available color, i.e., smallest color that keeps the  coloring up to these step harmonious. In this section we show that this greedy algorithm is a $\Omega(\sqrt{n})$-approximation even in the case of trees, where $n$ is the number of nodes in the tree. The result is stated in the next theorem.

\begin{theorem}
There exists a tree $T$ with $n = N(N-1)$ vertices that has a harmonious coloring with $2N - 2$ colors and is colored by the greedy algorithm with $(N-1)^2+1$ colors for a certain ordering of its vertices.
\end{theorem}

\begin{proof}
The tree $T$, illustrated in Figure~\ref{fig:counterex_greedy} is defined as follows. The root $a_0$ has $N-1$ children $a_1, \dots, a_{N-1}$. Each of the $N-2$ nodes $a_2, \dots, a_{N-1}$ has only one children. We term the children of the node $a_i$ with $b_i$. Then, each of the nodes $b_2, \dots b_{N-1}$ has $N-1$ children. We denote the $N-1$ children of the node $b_i$ as $c^1_{i}, c^2_{i}, \dots, c^{N-1}_{i}$. Tree $T$ has $n = N + N-2 + (N-1)(N-2) = N(N-1)$ vertices.

The greedy algorithm colors the root $a_0$ with $1$, $a_1$  with $2$, and the nodes $a_2, \dots, a_{N-1}$ with colors $3, 4, \dots, N$. Then, each of the nodes $b_2, \dots, b_{N-1}$ have color $2$. Finally, each of the nodes $c^j_{i}$ have a distinct color, which results in a total of $N + (N-1)(N-2)=(N-1)^2+1$ colors.

A coloring with $2N - 2$ is as follows. The root $a_0$ is colored with $1$ and the vertices $a_2, \dots, a_{N-1}$ with colors $2, 3, \dots, N$. In turn, the nodes $b_2, \dots, b_{N-1}$ are colored with colors $N+1, N+2, \dots, 2N - 2$. Finally, for every $2\le i\le N-1$ the nodes $c_i^j$ with $1\le j\le N-1$ are colored with the colors from the set $\{1,2,\dots,N\}$ different than the color of $a_i$.

Therefore, the greedy algorithm has an approximation factor of $\Omega(N) = \Omega(\sqrt{n})$.
  \end{proof}

\begin{figure}[h!tb]
  \includegraphics[scale=1] {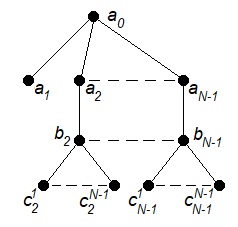}
  \centering
  \caption{Counterexample for the greedy algorithm}
  \label{fig:counterex_greedy}
\end{figure}

\subsection{Relating the harmonious coloring to the size of a minimum vertex cover}

In this subsection we show a connection between the harmonious chromatic number and the size of a minimum vertex cover of a graph. Before we present our theorem, we remind the definition of the vertex cover.

\begin{definition}
Given an undirected graph $G=(V,E)$, a subset $V' \subseteq V$ is a vertex cover of $G$ if for any edge $(a,b) \in E$ we have that either $a \in V'$, $b \in V'$ or both $a,b \in V'$.
\end{definition}

Finding the minimum vertex cover is a classical NP-hard problem for which the best approximation algorithm has a factor of 2 (see \cite{Vazirani10}). We now state our result.

\begin{theorem}
For any undirected graph $G = (V,E)$ we have that $h(G) \le VC + \Delta^2 - \Delta+1$, where $VC$ is the size of a minimum vertex cover of $G$ and $\Delta$ is the maximum degree of $G$.
\end{theorem}

\begin{proof}
Let $V'$ be a minimum vertex cover of $G$ and let $VC = |V'|$. For a vertex $x \in V$ let $N(x)$ be the open neighborhood of $x$, that is $N(x)= \{y \in V : (x,y) \in E\}$. For a set of vertices $X \subseteq V$, let $N(X)$ be set of all neighbours of vertices in $X$, that are not in $X$, that is $N(X) = \{b \in V - X : \exists a \in X \text{ s.t. } (a,b) \in E \}$.

We show a simple algorithm that colors any undirected graph with $VC + \Delta^2 - \Delta+1$ colors. First, we color the vertices of $V'$ with distinct colors from the set $\{1,2 \dots, VC \}$. Then, we process the vertices in $V - V'$, one by one in an arbitrary order and we color them as follows. For each vertex $x \in V - V'$ we simply assign one of the colors in the set $\{VC + 1, \dots, VC + \Delta^2 - \Delta +1\}$ such that the coloring remains harmonious.

 We now show that there always exist one such color in the set $\{VC + 1, \dots, VC + \Delta^2 - \Delta +1 \}$. More exactly, we prove that for a vertex $x \in V - V'$ there are at most $\Delta^2 - \Delta$
vertices in $V - V'$ that are at distance at most $2$ from $x$, hence require a color different that $x$.

Let $x \in V - V'$. First note that  $N(x) \subseteq V'$, since two vertices that are not in the vertex cover cannot be neighbours (otherwise, the definition of the vertex cover is violated). Consider now $N(N(x))$. Since the maximum degree of a vertex in $N(x)$ is $\Delta$ and each vertex in $N(x)$ is adjacent to $x$, it follows that there are at most $(\Delta-1)|N(x)|$ vertices in $N(N(x))-\{x\}$. But $|N(x)|\le \Delta$, hence there are at most $(\Delta-1)\Delta $ vertices in $V-V'$ that are at distance at most $2$ from $x$. 
\end{proof}

\section{Exact value of the harmonious chromatic number for some particular graphs, and classes of graphs}
\label{sec:particular_classes}

In this section we determine the harmonious chromatic number for some families of graphs like regular graphs and cycle-related graphs. We remind that the results for the graphs with diameter $2$ are presented in Section~\ref{sec:previous_particular}.

 \subsection{3-regular graphs of diameter 3}
First, recall the definition of a regular graph.

\begin{definition}
\label{def:regular-graph}
A connected graph $G$ is a regular graph if every vertex of $G$ has the same number of neighbors, so every vertex has the same degree. A regular graph with vertices of degree $r$ is called a $r‑regular$ graph or regular graph of degree $r$.
\end{definition}

Corollary~\ref{cor:diameter_2} refers to any 
graphs with diameter two, including $r$-regular graphs. Thus, the next statement follows. 

\begin{corollary}
\label{regular_2}
All $r$-regular graphs $G$ with $n$ vertices and diameter 2 have harmonious chromatic number $n$.
\end{corollary}


For example, octahedron is a $4$-regular graph with diameter $2$ (Figure~\ref{fig:octahedron}), Wagner graph (Figure~\ref{fig:wagner}) and Petersen graph (Figure~\ref{fig:petersen}) are $3$-regular graphs of diameter $2$, hence they have the harmonious chromatic number $n$. 

\begin{figure}[h!tb]\centering

\begin{minipage}[b]{0.3\linewidth}\includegraphics[scale=0.5]{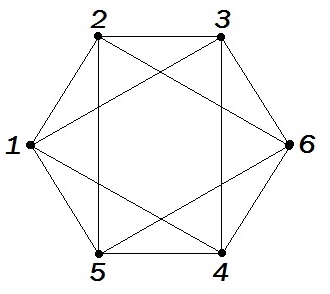}
  \centering
  \caption{A harmonious 6-coloring of octahedron graph}
  \label{fig:octahedron}
\end{minipage}
\quad
\begin{minipage}[b]{0.3\linewidth}\includegraphics[scale=0.62]{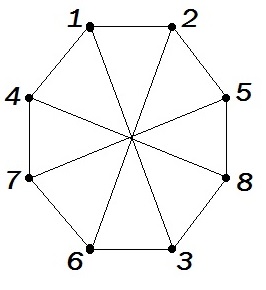}
  \centering
  \caption{A harmonious 8-coloring of Wagner graph}
  \label{fig:wagner}
  \end{minipage}
  \quad
\begin{minipage}[b]{0.3\linewidth}\includegraphics[scale=0.5] {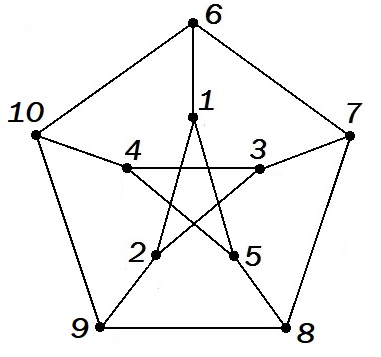}
  \centering
  \caption{A harmonious 10-coloring of Petersen graph}
  \label{fig:petersen}
  \end{minipage}
\end{figure}

A graph with maximum degree $\Delta$ and diameter $diam$ is called a $(\Delta, diam)$-graph. 
We determine the harmonious chromatic number for all $(3,3)$-regular planar graphs, and for well known $(3,3)$-regular non-planar graphs.  
\cite{all_3_regular} give a list of $3$-regular graphs of diameter $3$.
\begin{proposition}
\label{prop:min_7}
For a $(3,3)$-regular graph $G$ the minimum number of colors for a harmonious coloring is $7$.
\end{proposition}
\begin{proof}
Let $G(V,E)$ be a $(3,3)$-regular graph. Then, obviously, $|V(G)| \ge 8$. Let $c$ be a harmonious coloring of $G$. If all colors are distinct, then at least $8$ colors are used. Otherwise, there are two distinct vertices $v$, and $u$ with $c(u)=c(v)$. Then, vertices from $N(u)\cup N(v)$ must have distinct colors, different from $c(u)$. Since $c(u)=c(v)$, we have $d(u,v)\ge 3$ and  $N(u)\cap N(v)=\emptyset$. It follows that there are $6$ vertices in $N(u)\cup N(v)$, all having distinct colors, different from $c(u)$, hence at least $7$ colors are used. 
  \end{proof}

\subsubsection{Planar $(3,3)$-regular graphs}

\cite{pratt1996complete} establishes that the smallest $3$-regular planar graph with the
diameter $3$ has $8$ vertices and the largest  $3$-regular planar graph with the diameter $3$ has $12$ vertices. 
The number of non isomorphic planar $(3,3)$-regular graphs with $8$ vertices is $3$, with $10$ vertices is $6$, and with $12$ vertices is $2$. Note that an $r$-regular graph with $r$ odd must have an even number of vertices (Handshaking lemma). 

Figure~\ref{fig:graph_3_3} displays all the $(3,3)$-regular planar graphs with 8 vertices. Figure~\ref{fig:graph_3_3_10}  display all $(3,3)$-regular planar graphs with 10 vertices. Figure~\ref{fig:truncated_tetrahedron} display the two $(3,3)$-regular planar graphs with $12$ vertices. 

\begin{figure}[ht]
  \includegraphics[scale=0.6]{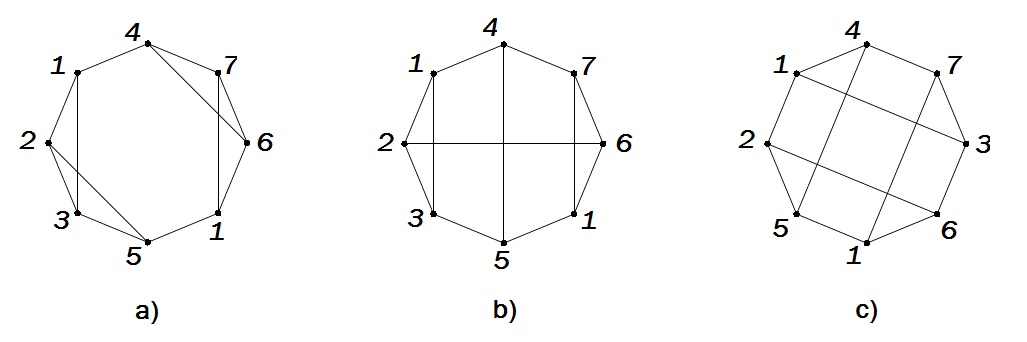}
  \centering
  \caption{Harmonious 7-colorings of  all (3,3)-regular planar graphs  with 8 vertices}
  \label{fig:graph_3_3}
\end{figure}

\begin{proposition}\label{harmonious_3_3_8}
For $(3,3)$-regular graphs with $8$ or $10$ vertices the harmonious chromatic number is 7.
\end{proposition}

 \begin{proof}
From Proposition~\ref{prop:min_7}, the number of colors for a harmonious coloring of a $(3,3)$-regular graph with $8$ or $10$ vertices is at least $7$. Then, to prove the result, it suffices to provide $7$-harmonious colorings for these graphs.
In Figure~\ref{fig:graph_3_3} we present all planar $(3,3)$-regular graphs with $8$ vertices along with a $7$-harmonious coloring of each of them and in Figure~\ref{fig:graph_3_3_10} we present harmonious colorings with $7$ colors for each planar $(3,3)$-regular graphs with $10$ vertices. 
   \end{proof}
 
\begin{figure}[ht]
  \includegraphics[scale=0.7] {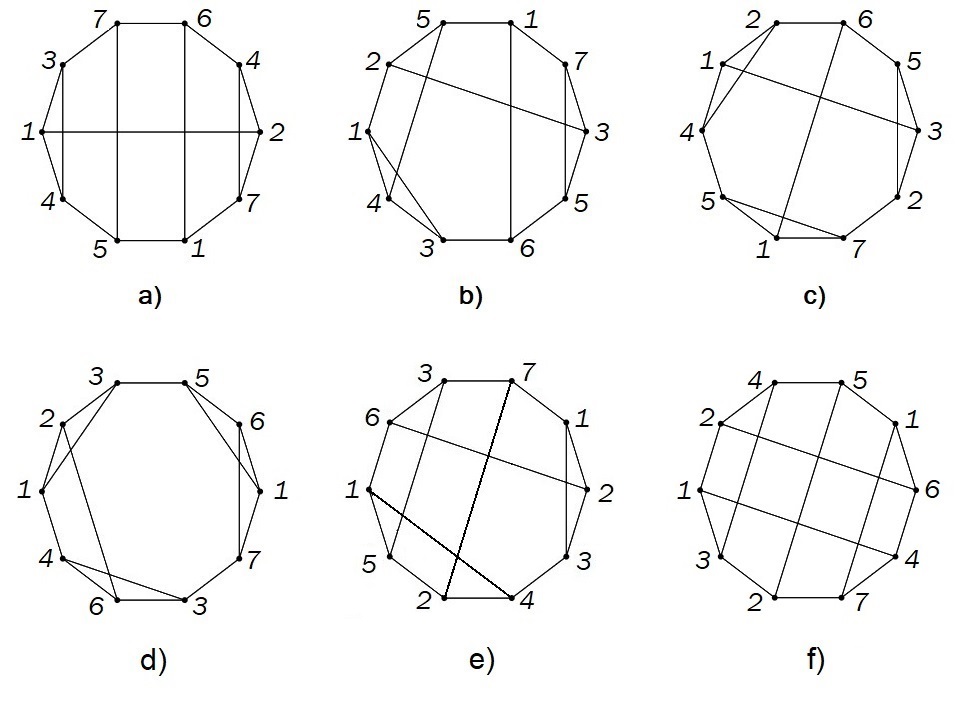}
  \centering
  \caption{Harmonious 7-colorings of all $(3,3)$-regular planar graphs with 10 vertices}
  \label{fig:graph_3_3_10}
\end{figure}

\begin{figure}[h!t]
\includegraphics[width=0.57 \textwidth]{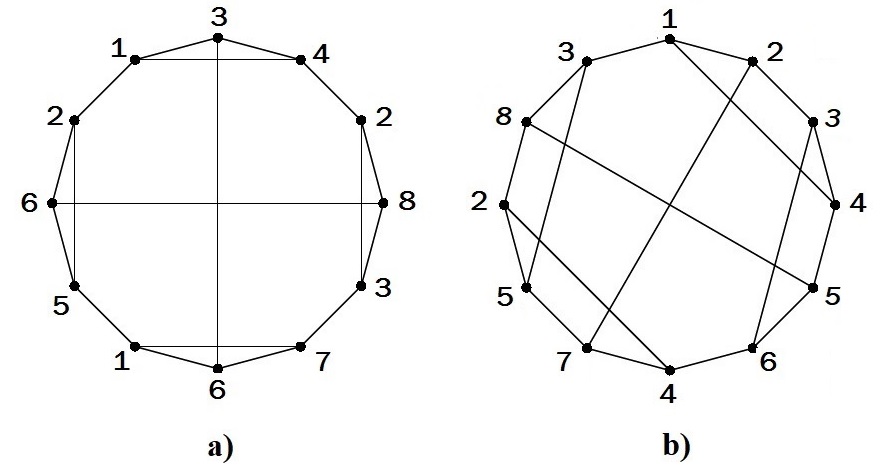}
    \centering
   \caption{Harmonious 8-coloring for the two (3,3)-regular planar graphs with 12 vertices} 
  \label{fig:truncated_tetrahedron}
\end{figure} 

\begin{theorem}\label{th12vf}
The harmonious chromatic number for the only 2 planar (3,3)-graphs with 12 vertices is 8.
\end{theorem}

\begin{proof}
Figure~\ref{fig:truncated_tetrahedron} shows a harmonious $8$-coloring of the truncated tetrahedron graph and a harmonious $8$-coloring of the second $(3,3)$-regular planar graph with $12$ vertices.

Using a computer program, we proved that these graphs cannot be colored harmoniously with less colors, by exhaustively trying all the possible harmonious colorings with $7$ colors. Our program is based on the classical backtracking schema: we color vertices one by one in increasing order of their index, and at one step we verify that there are no conflicts for the color $c$ assigned to the current vertex by considering the colors of the neighbours of all vertices previous colored with $c$.  
  \end{proof}

The source code of the program used in the proof of Theorem~\ref{th12vf} is available at~\cite{link_program}.

\subsubsection{Non-planar (3,3)-regular graphs}

In the previous section we determined the harmonious chromatic number for all $(3,3)$-regular planar graphs. 
It is natural to study the harmonious chromatic number for $(3,3)$-regular graphs with $10$ or $12$ vertices that are no longer planar. A list of these graphs (described via their adjacency lists) can be found at~\cite{link_33}.

Using the program described in the proof of Theorem~\ref{th12vf}, we find that all  non-planar $(3,3)$-regular graphs with $10$ vertices  have harmonious chromatic number between $7$ and $9$ ($7$ in the planar case). One well-known non-planar $(3,3)$-regular graphs with $10$ vertices, the pentagonal prism graph $GP_{5,1}$, has $h(GP_{5,1})=7$ (Figure~\ref{fig:pentagonal}). Figure~\ref{fig:doi_10vertices} shows two other $(3,3)$-regular non-planar graphs with 10 vertices, one with $h=8$, and one with $h=9$ (this is the only graph of this type that has $h=9$).
\begin{figure}[ht]\centering
\begin{minipage}[b]{0.35\linewidth}\includegraphics[scale=0.55]{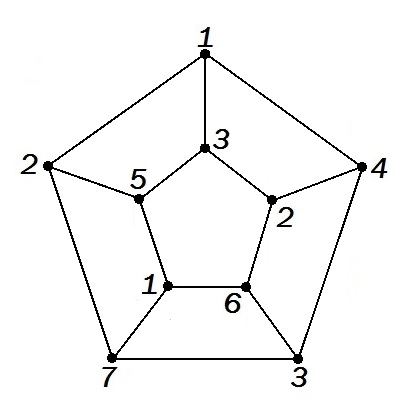}
\centering
\caption{A harmonious 7-coloring of $GP_{5,1}$}\label{fig:pentagonal}\end{minipage}
\quad
\begin{minipage}[b]{0.55\linewidth}\includegraphics[width=1.1 \textwidth]{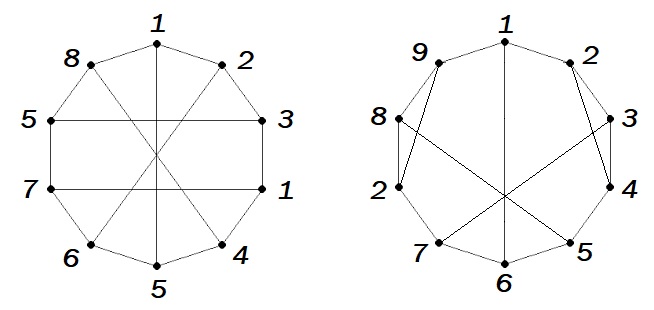}
\centering
\caption{Harmonious coloring of two $3$-regular graphs with 10 vertices}\label{fig:doi_10vertices}\end{minipage}
\end{figure}

According to~\cite{pratt1996complete}, there are $32$ $(3,3)$-regular non-planar graphs with $12$ vertices, among  which Franklin graph (Figure~\ref{fig:franklin}), Yutsis graph (Figure~\ref{fig:Yutsis}), and Tietze's graph (obtained from Petersen graph by expanding one vertex to a triangle;  Figure~\ref{fig:tietze}), all having harmonious chromatic number $h=9$, and Bidiakis graph (Figure~\ref{fig:bidiakis}), with $h=8$. 

\begin{figure}[ht]\centering
\begin{minipage}[b]{0.45\linewidth}\includegraphics[scale=0.54]{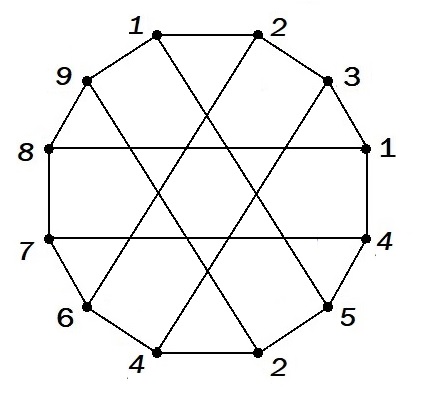}
\centering
\caption{A harmonious 9-coloring of Franklin graph}\label{fig:franklin}\end{minipage}
\quad
\begin{minipage}[b]{0.45\linewidth}\includegraphics[width=0.62 \textwidth]{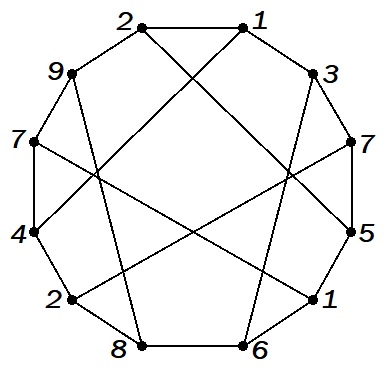}
\centering
\caption{A harmonious 9-coloring of Yutsis graph}\label{fig:Yutsis}\end{minipage}
\quad
\begin{minipage}[b]{0.42\linewidth}
\includegraphics[scale=0.5]{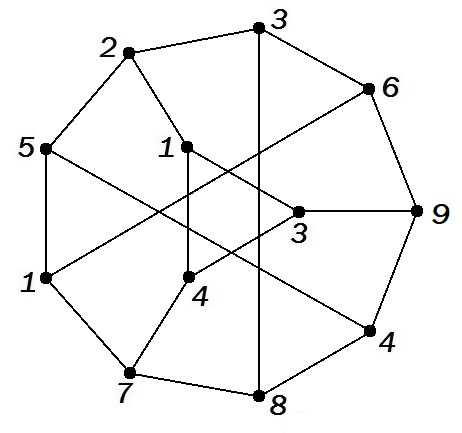} 
\centering \caption{A harmonious 9-coloring of Tietze graph} 
\label{fig:tietze}
\end{minipage}
\quad
\begin{minipage}[b]{0.45\linewidth}\includegraphics[scale=0.55]{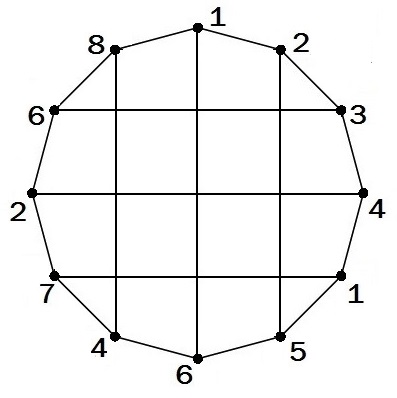}
\centering \caption{A harmonious 8-coloring of Bidiakis graph}\label{fig:bidiakis}\end{minipage}
\end{figure}

The $(3,3)$-regular non-planar graphs can have more than $12$ vertices. Although we could not classify all $(3,3)$-regular graphs according to their harmonious chromatic number, we fully explore the planar graphs from this category and provide a tool - a computer program - to explore the harmonious coloring of these graphs when the number of vertices is small enough. The harmonious chromatic number for other interesting $3$-regular graphs up to $24$ vertices obtained using our computer program  can be found in ~\cite{link_program}.

\subsection{Some families of cycle-related graphs}

In previous work are determined the value of the harmonious chromatic number for some graphs generated from a cycle, like wheel graph $W_n$
, gear graph $G_n$
, and Helm graph $H_n$ 
by~\cite{huilgol2016harmonious}.  These graphs have $n$ vertices on a cycle connected to a central vertex, and then $\Delta=n$. The harmonious chromatic number is $h(W_n)=h(G_n)=h(H_n)=n+1$.

Next, we determine de exact values for the harmonious chromatic number of other families of cycle-related graphs, like: sunflower graph, flower graph, double wheel graph, sun graph, closed sun graph, and lollipop graph.

\subsubsection{Cycle-related graphs with diameter $2$}

The are several interesting 
cycle-related graphs with diameter $2$, like the followings. Each of these graphs has $2n+1$ vertices, and $\Delta=2n$. Thus, from Corollary~\ref{cor:diameter_2}, the harmonious chromatic number  is equal with the number of their vertices, $h= 2n+1$.
\begin{itemize}
\item Flower graph  $Fl_n$, obtained from Helm graph $H_n$  by joining every pendant with the central vertex (see Figure~\ref{fig:Flower}).

\item Double wheel graph $W_{n,n}$, 
 obtained from two wheel graphs $W_n$ sharing the same universal vertex (also called center) $v_0$, with vertices on the cycle denoted $v_1,v_2,\dots,v_n$, respectively $u_1,u_2,\dots,u_n$ (see Figure~\ref{fig:double_wheel}).

\item The graph 
{$G_{n,n}$},  
 obtained from $W_{n,n}$ by connecting $v_i$ with $u_i$, where $1\leq i\leq n$ (see Figure~\ref{fig:new_prism}). 
 
 \item Triangular book $B_{3,n}$, defined by set of vertices $V=\{v,u,v_i: 1\leq i\leq n\}$ and set of edges $E=\{uv,uv_i,vv_i: 1\leq i\leq n\}$ (see Figure~\ref{fig:book}).

 \item Triangular book with bookmark $TB_{3,n}$, defined by set of vertices $V=\{v,u,x,v_i: 1\leq i\leq n\}$ and set of edges $E=\{uv,ux,uv_i,vv_i: 1\leq i\leq n\}$ (see Figure~\ref{fig:book_mark}). 
 
 \item Jewel graph $J_n$, defined by the set of vertices $V=\{u,v,x,y,v_i:1\leq i \leq n\}$, and the set of edges $E=\{ux,xv,xy,uy,yv,uv_i, vv_i:1\leq i \leq n\}$ (see Figure~\ref{fig:jewel}). 
\end{itemize}

\begin{figure}[ht]\centering
\begin{minipage}[b]{0.3\linewidth}
\includegraphics[width=0.9\textwidth]{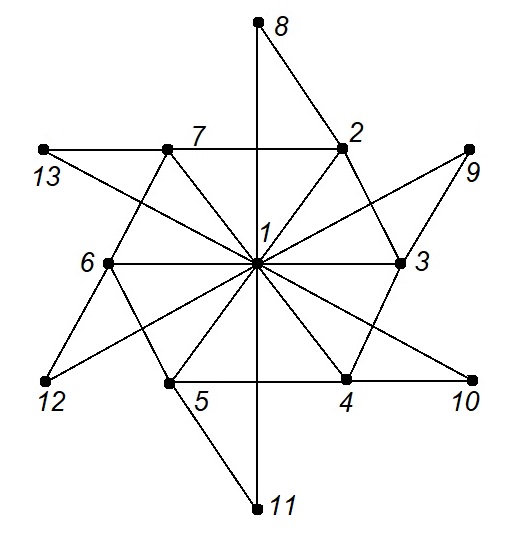}
\centering
\caption{A harmonious 13-coloring of flower graph $Fl_6$}\label{fig:Flower}
\end{minipage}
\quad
\begin{minipage}[b]{0.3\linewidth}\includegraphics[width=0.9 \textwidth]{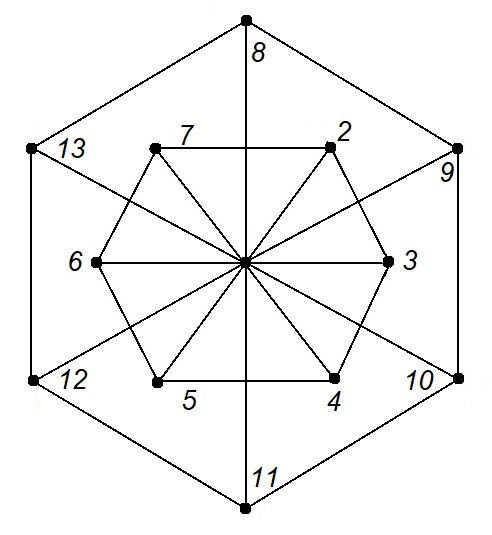}
\centering\caption{A harmonious 13-coloring of double wheel $W_{6,6}$}\label{fig:double_wheel}\end{minipage}
\quad
\begin{minipage}[b]{0.3\linewidth}\includegraphics[width=0.9 \textwidth]{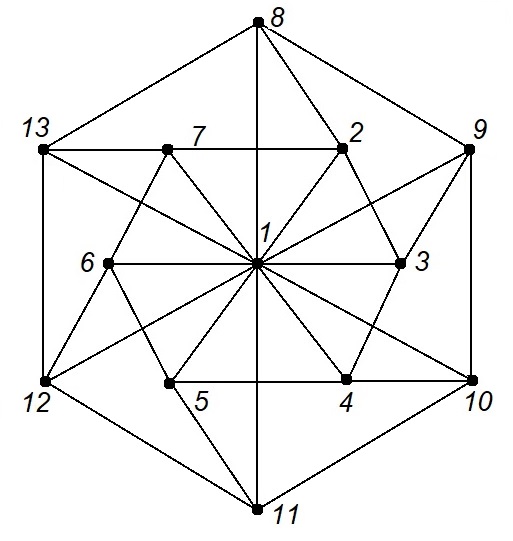}
\centering\caption{ A harmonious 13-coloring of $G_{n,n}$}\label{fig:new_prism}\end{minipage}
\end{figure}

\begin{figure}[ht]\centering
\begin{minipage}[b]{0.3\linewidth}
\includegraphics[width=0.7\textwidth]{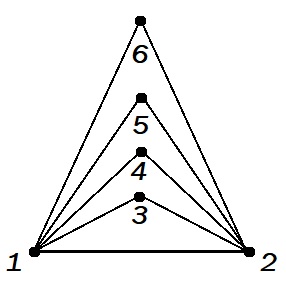}
\centering
\caption{A harmonious 6-coloring of triangular book graph $B_{3,4}$}\label{fig:book}
\end{minipage}
\quad
\begin{minipage}[b]{0.3\linewidth}\includegraphics[width=0.7 \textwidth]{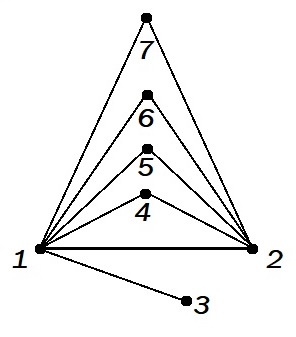}
\centering\caption{A harmonious 7-coloring of triangular book  with bookmark $TB_{3,4}$}\label{fig:book_mark}\end{minipage}
\quad
\begin{minipage}[b]{0.3\linewidth}\includegraphics[width=0.7 \textwidth]{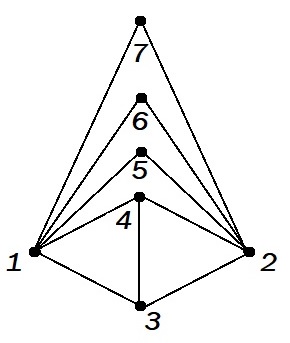}
\centering\caption{ A harmonious 7-coloring of jewel graph $J_3$}\label{fig:jewel}\end{minipage}
\end{figure}

\subsubsection{Cycle-related graphs with diameter greater than 2}

Next we consider four cycle-related graphs with diameter  greater than $2$.

\begin{enumerate}
\item The \emph{sunflower graph} $Sf_n$ is obtained from a $n$-wheel graph $W_n$ with set of vertices $\{v_0,v_1,v_2,\dots,v_n\}$ by adding $n$ vertices $u_i$, $1\leq i \leq n$, and joining each new vertex $u_i$ with two adjacent vertices $v_i$, $v_{i+1}$, $1\leq\ i \leq n-1$, and $u_n$ with $v_{n}$ and $v_1$. Thus, $Sf_n$ has $2n+1$ vertices, and $4n$ edges. The degree for each vertex of $Sf_n$: $d(v_0)=n$, $d(v_i)=5$, and $d(u_i)=2$, where $1 \leq i \leq n$.

\begin{theorem}\label{h_sunflower_graph}
The sunflower graph ${Sf}_n$ has $h(Sf_n)=7$, for $3 \leq n \leq 4$, $h(Sf_n)=8$ for $5 \leq n \leq 6$, and  $h(Sf_n)=n+1$, for $n\geq 7$.
\end{theorem}
\begin{figure}[h!tb]\centering
\begin{minipage}[b]{0.3\linewidth}
  \includegraphics[scale=0.45] {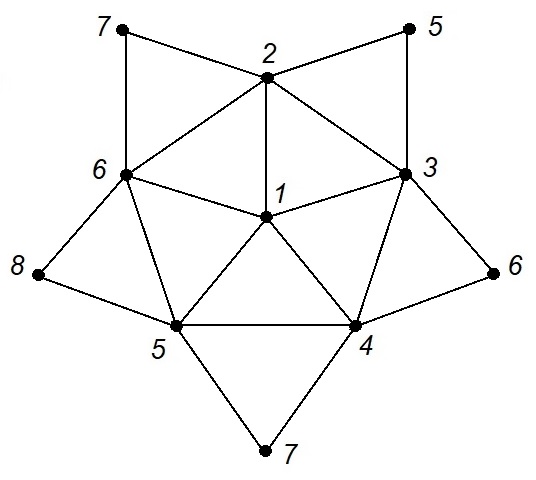}
  \centering
  \caption{A harmonious 8-coloring of $Sf_5$}
  \label{fig:sunflower_5}
\end{minipage}
\quad
\begin{minipage}[b]{0.3\linewidth}
  \includegraphics[scale=0.42] {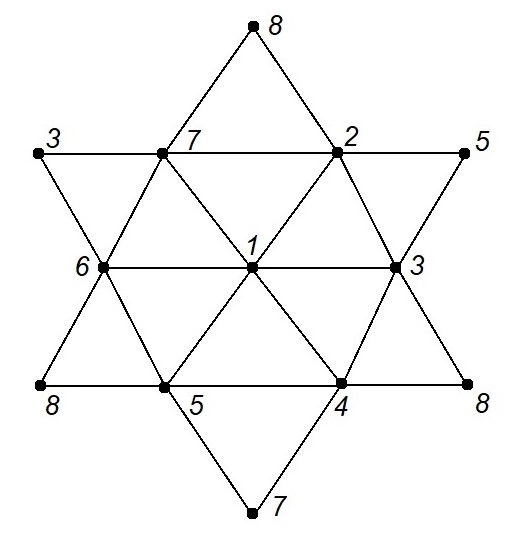}
  \centering
  \caption{A harmonious 8-coloring of $Sf_6$}
  \label{fig:sunflower_6}
\end{minipage}
\quad
\begin{minipage}[b]{0.3\linewidth}
  \includegraphics[scale=0.42] {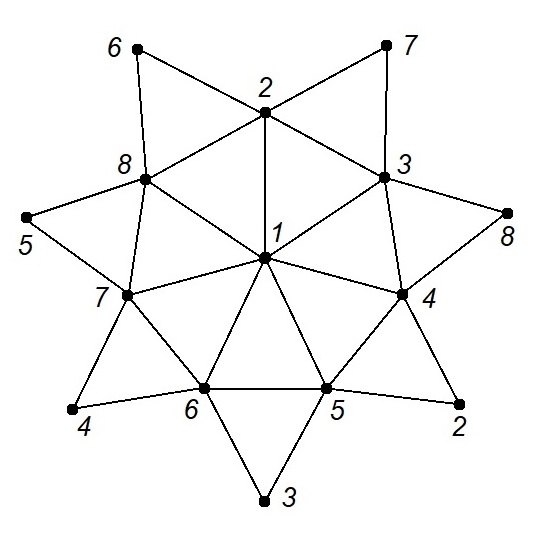}
  \centering
  \caption{A harmonious 8-coloring of $Sf_7$}
  \label{fig:sunflower_7}
  \end{minipage}
\end{figure}
\begin{proof} The sunflower graph $Sf_3$ has diameter $2$, and thus, for Corollary~\ref{cor:diameter_2}, $h(Sf_3)=7$. For $4\le n\le 6$ we used our computer program described in proof of Theorem \ref{th12vf} to obtain the harmonious chromatic number (see Figure~\ref{fig:sunflower_5} and Figure~\ref{fig:sunflower_6}).

The sunflower graph $Sf_n$, with $n \geq 7$, has the harmonious chromatic number $h(Sf_n)\ge h(W_n)=n+1$. In order to prove that equality holds, we describe a harmonious coloring for $Sf_n$ with $n+1$ colors.

Color the central vertex $v_0$ with color $1$; then color the vertices $v_i$, $1\leq i\leq n$, on the cycle with colors in order in set $C=\{2,3,\dots,{n+1}\}$, clockwise, and assign to the vertices $u_i$ colors in set $C$, clockwise, starting from the vertex $u_2$, situated at distance 3 from the vertex $v_1$ previously colored with $2$ (Figure~\ref{fig:sunflower_7}). 
  \end{proof}

\item The \emph{sun graph} $S_n$ is obtained from the complete graph $K_n$, with vertices denoted $v_1,v_2,\dots,v_n$ and $n$ new vertices $u_1,u_2,\dots,u_n$, each connected with two adjacent vertices on an outer cycle of $K_n$, more precisely vertex $u_i$ is adjacent with $v_i$ and $v_{i+1}$, for every $1\leq i\leq n-1$, and $u_n$ is adjacent with $v_n$ and $v_1$. Thus, the sun graph $S_n$ has $2n$ vertices, and $n(n-1)/2+2n$ edges.

\begin{theorem}\label{h_sun_graph} The sun graph $S_n$, $n\ge 3$ has $h(S_n)=n+2$ if $n$ is even, and  $h(S_n)=n+3$ if $n$ is odd. \end{theorem}
\begin{figure}[h!tb]\centering
\begin{minipage}[b]{0.45\linewidth}
  \includegraphics[scale=0.45] {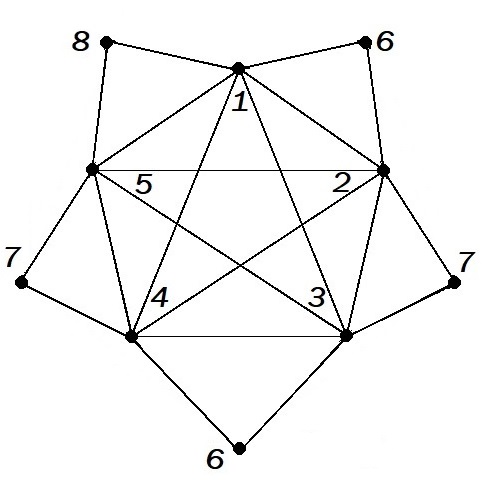}
  \centering
  \caption{A harmonious 8-coloring of $S_5$ }
  \label{fig:sun_5}
\end{minipage}
\quad
\begin{minipage}[b]{0.45\linewidth}
  \includegraphics[scale=0.45] {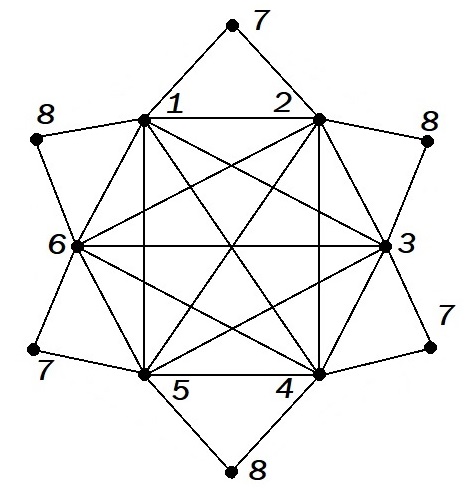}
  \centering
  \caption{A harmonious 8-coloring of $S_6$ }
  \label{fig:sun_6}
  \end{minipage}
\end{figure}
\begin{proof} In a harmonious coloring of $S_n$ vertices $v_1,\ldots,v_n$ of the clique must have distinct colors. Denote these colors $1,\ldots,n$. Since $d(u_i,v_j)\le 2$ for every $1\le i,j\le n$, it follows that colors $1,\ldots,n$ cannot be used for vertices $u_1,\ldots, u_n$. Moreover, since $d(u_i,u_{i+1})=2$ for every $1\le i\le n-1$ and $d(u_n,u_1)=2$, it follows that, if $n$ is even at least $2$ new colors are needed for vertices $u_1,\ldots,u_n$ and if $n$ is odd at least $3$ new colors are needed.
Hence
$$
h(S_n)\ge\left\{
\begin{array}{ll} 
n+2, &\mbox{if $n$ even}\\
n+3, &\mbox{if $n$ odd.}\\
\end{array}\right.
$$
The lower bound can be achieved for the following coloring, hence equality holds:

\begin{itemize}
    \item for $n$ even, let $c(v_i)=i$ for every $1\le i\le n$, $c(u_j)=n+1$ if $j$ is odd and $c(u_j)=n+2$ if $j$ is even for $1\le j\le n$ (Figure~\ref{fig:sun_6}),
    \item for $n$ odd, let $c(v_i)=i$ for every $1\le i\le n$, $c(u_j)=n+1$ if $j$ is odd and $c(u_j)=n+2$ if $j$ is even for $1\le j\le n-1$ and $c(u_n)=n+3$ (Figure~\ref{fig:sun_5}).
\end{itemize}
 
\end{proof}

\item The \emph{closed sun graph}  $\overline{\rm S_n}$ is the graph $S_n$ with edges between vertices $u_i, u_{i+1}$, where $1\leq i<n$, and between $u_n$ and $u_1$. Thus, $\overline{\rm S_n}$ has 2n vertices and $n(n-1)/2+3n$ edges. Then, $d(v_i)=n+1$, and $d(u_i)=4$.

\begin{figure}[h!tb]\centering
\begin{minipage}[b]{0.45\linewidth}
  \includegraphics[scale=0.45] {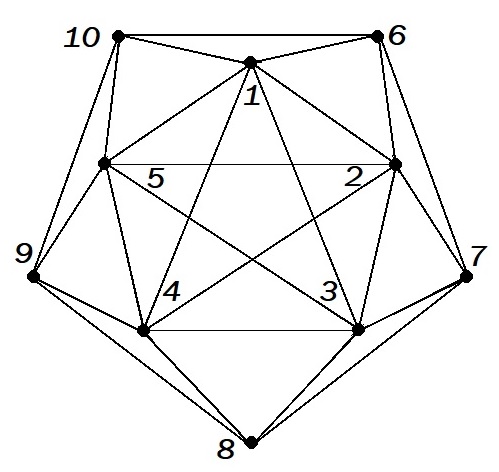}
  \centering
  \caption{A harmonious 10-coloring of 
  $\overline{\rm S_5}$}
  \label{fig:closed_sun_5}
\end{minipage}
\quad
\begin{minipage}[b]{0.45\linewidth}
  \includegraphics[scale=0.45] {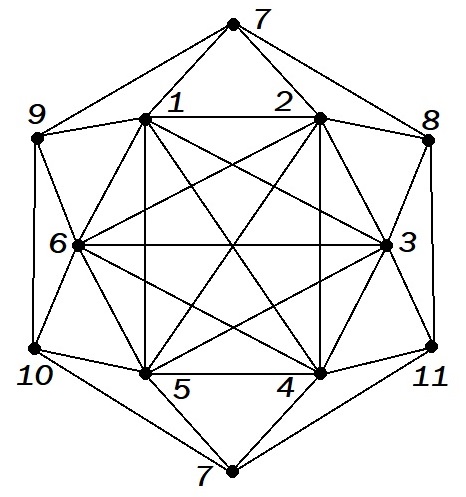}
  \centering
  \caption{A harmonious 11-coloring of 
  $\overline{\rm S_6}$ }
  \label{fig:closed_sun_6}
  \end{minipage}
\end{figure}

\begin{theorem}\label{h_closed_sun_graph}
The closed sun graph $\overline{\rm S_n}$ has $h(\overline{\rm S_n})=2n$, for $n\le 5$  and  $h(\overline{\rm S_n})=n+h(C_n)$, for $n > 5$.
\end{theorem}

\begin{proof}
For $n\le 5$ we have $h(\overline{\rm S_n})=2n$, since in this case $\overline{\rm S_n}$ has diameter $2$.
For $n>5$, vertices of the clique must be colored with $n$ distinct colors and these colors cannot be used for any vertex from the outer cycle $C_n$, since a vertex from the outer cycle is at distance at most $2$ from any vertex of the clique; hence we have $h(\overline{\rm S_n})\ge n+h(C_n)$.    
To prove that equality holds, we consider the following coloring for $\overline{\rm S_n}$ (Figure~\ref{fig:closed_sun_5}, Figure~\ref{fig:closed_sun_6}), which can be easily verified that is harmonious:

\begin{itemize}
    \item first color with $1,\ldots,n$ the vertices of the clique $K_n$,
    \item then consider a harmonious coloring for the outer cycle $C_n$ with $h(C_n)$ colors, using colors from $n+1$ to $n+h(C_n)$.
\end{itemize}

\end{proof}


\item The \emph{Lollipop graph} $L_{n,m}$

Let $G$, $H$ be two connected graphs and consider one vertex from each of these two graphs: $a\in V(G)$, $b\in V(H)$. Denote by $(G,a) \odot (H,b)$ the graph obtained from the union of graphs $G$ and $H$ by identifying vertices $a$ and $b$. We will call this operation vertex-union.

For two positive numbers $n\ge 3$, $m\ge 2$ \emph{Lollipop graph} $L_{n,m}$ is the vertex-union $(K_n,u)\odot (P_m,v)$ where $u$ is any vertex of a clique $K_n$ and $v$ is a degree $1$ vertex of path $P_m$. 

\begin{figure}[h!tb]  \includegraphics[scale=0.65] {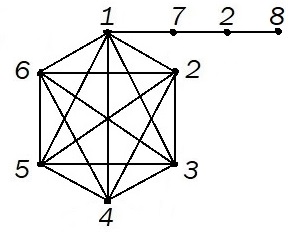}  \centering  \caption{A harmonious 8-coloring of lollipop graph $L_{6,4}$}  \label{fig:lollipop}\end{figure}

\begin{theorem}\label{h_Lollipop}
Let $n\ge 3 $ and $m\ge 2$ and $t$ be the minimum natural number such that $m \le 1+nt+\frac{t(t-1)}{2}$.
The Lollipop graph $L_{n,m}$ has the harmonious chromatic number $h(L_{n,m})=n+t$ in the following cases:
\begin{itemize}
    \item $t$ is even and $n$ is odd
    \item $t$ and $n$ are even and $m \le 1+ nt+\frac{t(t-1)}{2} - \frac{t}{2}$
    \item $t$ is odd, $n$ is even and and $m\le 1+ nt+\frac{t(t-1)}{2} - (n-2)$
      \item $t$ and $n$ are odd  and $m\le 1+ nt+\frac{t(t-1)}{2}-(n-2+\max({\frac{t-(n-2)}{2},0})))$
\end{itemize}
otherwise  $h(L_{n,m})=n+t+1$.

\end{theorem}

\begin{proof}
In this proof for a complete graph $K_r$ we denote the vertices with $1,\ldots,r$. Also, for $n\le r$ we denote by $\left<[n]\right>$ the clique induced in $K_r$ by vertices $1,\ldots,n$.

Let $k=nt+\frac{t(t-1)}{2}$.

Let $r=h(L_{n,m})$ and let $c$ be an $r$-harmonious coloring of $L_{n,m}$. The $n$ vertices of the clique of $L_{n,m}$ must have distinct colors: assume w.l.o.g, that these colors are $1,\ldots,n$. 
Then, in the complete graph $K_r$, according to the coloring $c$, the colors of the vertices of the clique in $L_{n,m}$ correspond to a clique with $n$ vertices $1,\ldots,n$ in $K_r$ and the colors of the vertices from the  path $P_m$ of  $L_{n,m}$ correspond to a trail (possible closed) with $m$ vertices in $K_r-E(\left<[n]\right>)$ (obtained from $K_r$ by removing all the edges between vertices $1,\ldots,n$) starting with a vertex from $1,\ldots,n$. 

Conversely, if in a clique $K_r$ with $r\ge n$ there exists a trail with $m$ vertices in $K_r-E(\left<[n]\right>)$ starting with a vertex from $1,\ldots,n$ (assume w.l.o.g. it starts from vertex $1$), then $L_{n,m}$ has a $r$-harmonious coloring. It follows that the harmonious chromatic number of $L_{n,m}$ is the minimum $r$ with such property.

Let $t$ be the smallest number such that $|E(L_{n,m})|=|E(K_n)|+|E(P_m)|\le E(K_{n+t})$, that is such $m-1\le nt+\frac{t(t-1)}{2}=k$. Then $h(L(K_{n,m}))\ge n+t$ and equality holds only if  the following property is satisfied: there exists a trail with $m$ vertices in $K_{n+t}-E(\left<[n]\right>)$ starting with a vertex $1$.

In $K_{n+t}-E(\left<[n]\right>)$ vertices $1,\ldots,n$ have degree $t$ and vertices $n+1,\ldots, n+t$ have degree $n+t-1$. In order to have a trail with $m$ vertices starting with vertex $1$ in  $K_{n+t}-E(\left<[n]\right>)$, the largest subgraph of this graph that has an Eulerian trail must have at least $m-1$ edges and all vertices of this subgraph must have even degree with at most $2$  exceptions; if there are vertices of odd degree in this subgraph, then vertex $1$ must be one of them, thus at least $t-1$ of vertices $n+1,\ldots,n+t$ have even degree in this subgraph.

We consider four cases, according to the parity of $n$ and $m$.

\begin{enumerate}

\item[Case $1.$] If $t$ is even and $n$ is odd, then $K_{n+t}-E(\left<[n]\right>)$ is Eulerian, hence it has an Eulerian cycle. This cycle includes a trail with $m$ vertices starting from vertex $1$, hence in this case $h(L(K_{n,m}))= n+t$.

\item[Case $2.$] If $t$ is even and $n$ is even,  in order to have a subgraph in  $K_{n+t}-E(\left<[n]\right>)$ with all vertices from $n+1$ to $n+t$ of even degree with at most one exception, then we must remove at least $\frac{t}{2}$ edges, hence $m-1$ must be at most $k-\frac{t}{2}$. We can obtain such a subgraph by removing edges $(n+1,n+2)$, $(n+3,n+4),\ldots, (n+t-1,n+t)$. This subgraph is Eulerian, hence is has a trail with $m$ vertices starting from vertex $1$. It follows that if $m-1\le k-\frac{t}{2}$, then we have $h(L(K_{n,m})= n+t$. Otherwise, $h(L(K_{n,m}))\ge n+t+1$ and equality holds, since by adding a new vertex to $K_{n+t}$ and joining it with $n+1,\ldots,n+t$ we obtain an Eulerian subgraph of $K_{n+t+1}-E(\left<[n]\right>)$ with at least $m$ edges.

\item[Case $3.$] If $t$ is odd and $n$ is even, then, in order to have a subgraph with an Eulerian trail, we must remove edges such that at least $n-2$ vertices from $1,\ldots,n$ have even degree. Since these vertices are pairwise nonadjacent, we must remove at least $n-2$ edges. For example if we remove  $(3,n+1),(4,n+1),\ldots(n,n+1)$ we obtain a subgraph with an Eulerian trail with one extremity in $1$.

Hence, in this case, if $m-1\le k-(n-2)$, then we have $h(L(K_{n,m}))= n+t$, otherwise, as in Case $2$, $h(L(K_{n,m}))= n+t+1$.

\item[Case $4.$] If $t$ is odd and $n$ is odd, consider two subcases.

\begin{enumerate}

\item[Subcase $4.1$.] If $t\ge n-2$, as in Case $2$, we must remove at least $\frac{n+t}{2}-1$ edges in order to have a subgraph with an Eulerian trail. We can remove the edges: $(i,n+i-2)$ for $3\le i\le n$, and $(n+n-1,n+n),(n+n+1,n+n+2),\ldots,(n+t-1,n+t)$ and obtain the desired subgraph, hence in this case if $m-1\le k-(\frac{n+t}{2}-1)=k-(n-2+\frac{t-n+2}{2})$ we have $h(L(K_{n,m}))= n+t$, otherwise, as in  Case $2$, we have $h(L(K_{n,m}))= n+t+1$ .

\item[Subcase $4.2$.] If $t< n-2$, as in Case $3$,  we must remove at least $n-2$ edges such that at least $n-2$ vertices from $1,\ldots,n$ became of even degree. For example, we remove the edges $(i,n+i-2)$ for $3\le i\le t+1$ and the edges $(i,n+t)$ for $t+2\le i\le n$ and obtain a subgraph with an Eulerian trail from vertex $1$.
Hence, in this case, if $m-1\le k-(n-2)$, then we have $h(L(K_{n,m}))= n+t$, otherwise $h(L(K_{n,m})= n+t+1$.

\end{enumerate}

\end{enumerate}
  
\end{proof}

\end{enumerate}

\section{Conclusions and future work}
\label{sec:conclusions}

In this paper we studied the harmonious chromatic number, which is a proper vertex coloring such that for every two distinct colors $i$, $j$ at most one pair of adjacent vertices are colored with $i$ and $j$. 

We showed that finding a minimum harmonious colorings for arbitrary graphs is APX-hard, the natural greedy algorithm is a $\Omega(\sqrt{n})$-approximation, and, moreover, we show a relationship between the minimum vertex cover and the harmonious chromatic number. 
In the second part of our paper we determined the exact value of the harmonious chromatic number for all  $3$-regular planar graphs of diameter $3$, some non-planar regular graphs and cycle-related graphs.

We state an open problem related to the approximability of the harmonious chromatic number.

\begin{open}
Does there exist a constant factor approximation algorithm for the harmonious chromatic number on arbitrary graphs?
\end{open}

Finally, we list a couple of classes of cycle-related graphs for which it is interesting to find the exact value of the harmonious chromatic number: square graph, tadpole or dragon graph, barbell graph, diamond snake, total graph of path, total graph of cycle.

\bibliographystyle{abbrvnat}
\bibliography{bibliography}

\begin{thebibliography}{34}
\providecommand{\natexlab}[1]{#1}
\providecommand{\url}[1]{\texttt{#1}}
\expandafter\ifx\csname urlstyle\endcsname\relax
  \providecommand{\doi}[1]{doi: #1}\else
  \providecommand{\doi}{doi: \begingroup \urlstyle{rm}\Url}\fi

\bibitem[Asdre and Nikolopoulos(2007)]{asdre_2007_NP_classes}
K.~Asdre and S.~D. Nikolopoulos.
\newblock Np-completeness results for some problems on subclasses of bipartite
  and chordal graphs.
\newblock \emph{Theoretical Computer Science}, 381:\penalty0 248--259, 2007.
\newblock ISSN 0304-3975.
\newblock \doi{https://doi.org/10.1016/j.tcs.2007.05.012.}

\bibitem[Bodlaender(1989)]{bodlander1989cograph}
H.~L. Bodlaender.
\newblock Achromatic number is np-complete for cographs and interval graphs.
\newblock \emph{Information Processing Letters}, 31:\penalty0 135--138, 1989.
\newblock ISSN 0020-0190.
\newblock \doi{10.1016/0020-0190-89-90221-4}.

\bibitem[Dinur and Safra(2004)]{Dinur04onthe}
I.~Dinur and S.~Safra.
\newblock On the hardness of approximating minimum vertex cover.
\newblock \emph{Annals of Mathematics}, 162:\penalty0 2005, 2004.

\bibitem[Edwards(1996)]{edwards_1996}
K.~Edwards.
\newblock The harmonious chromatic number of bounded degree trees.
\newblock \emph{Combinatorics, Probability and Computing}, 5\penalty0
  (1):\penalty0 15–28, 1996.
\newblock \doi{10.1017/S0963548300001802}.

\bibitem[Edwards(1997)]{edwards_1997}
K.~Edwards.
\newblock \emph{The Harmonious Chromatic Number and the Achromatic Number},
  page 13–48.
\newblock London Mathematical Society Lecture Note Series. Cambridge University
  Press, 1997.
\newblock \doi{10.1017/CBO9780511662119.003}.

\bibitem[Edwards and McDiarmid(1995)]{edwards1995complexity}
K.~Edwards and C.~McDiarmid.
\newblock The complexity of harmonious colouring for trees.
\newblock \emph{Discrete Applied Mathematics}, 57\penalty0 (2-3):\penalty0
  133--144, 1995.

\bibitem[Frank et~al.(1982)Frank, Harary, and Plantholt]{frank1982line}
O.~Frank, F.~Harary, and M.~Plantholt.
\newblock The line-distinguishing chromatic number of a graph.
\newblock \emph{Ars Combin}, 14:\penalty0 241--252, 1982.

\bibitem[Graham and Sloane(1980)]{GrahamSloane}
R.~L. Graham and N.~J.~A. Sloane.
\newblock On additive bases and harmonious graphs.
\newblock \emph{SIAM Journal on Algebraic Discrete Methods}, 1\penalty0
  (4):\penalty0 382--404, 1980.
\newblock \doi{10.1137/0601045}.
\newblock URL \url{https://doi.org/10.1137/0601045}.

\bibitem[Hopcroft and Krishnamoorthy(1983{\natexlab{a}})]{complete}
J.~Hopcroft and M.~Krishnamoorthy.
\newblock On the harmonious coloring of graphs.
\newblock \emph{Siam Journal on Algebraic and Discrete Methods}, 4, 09
  1983{\natexlab{a}}.
\newblock \doi{10.1137/0604032}.

\bibitem[Hopcroft and Krishnamoorthy(1983{\natexlab{b}})]{hopcroft1983}
J.~E. Hopcroft and M.~S. Krishnamoorthy.
\newblock On the harmonious coloring of graphs.
\newblock \emph{Siam Journal on Algebraic Discrete Methods}, 4\penalty0
  (3):\penalty0 306--311, 1983{\natexlab{b}}.

\bibitem[Huilgol and Sriram(2016)]{huilgol2016harmonious}
M.~I. Huilgol and V.~Sriram.
\newblock On the harmonious coloring of certain class of graphs.
\newblock \emph{Journal of Combinatorics, Information \& System Sciences},
  41\penalty0 (1-3):\penalty0 17, 2016.

\bibitem[Ioannidou and Nikolopoulos(2010)]{subclasses_2010}
K.~Ioannidou and S.~Nikolopoulos.
\newblock Harmonious coloring on subclasses of colinear graphs.
\newblock In \emph{WALCOM: Algorithms and Computation. WALCOM 2010. Lecture
  Notes in Computer Science}, volume 5942. Springer, Berlin, Heidelberg, 2010.
\newblock \doi{https://doi.org/10.1007/978-3-642-11440-3_13}.

\bibitem[Katerina~Asdre and Nikolopoulos(2007)]{interval_permutation}
K.~I. Katerina~Asdre and S.~D. Nikolopoulos.
\newblock The harmonious coloring problem is np-complete for interval and
  permutation graphs.
\newblock \emph{Discrete Applied Mathematics}, 155:\penalty0 2377--2382, 2007.
\newblock ISSN 0166-218X.
\newblock \doi{https://doi.org/10.1016/j.dam.2007.07.005.}
\newblock URL
  \url{http://www.sciencedirect.com/science/article/pii/S0166218X0700251X}.

\bibitem[Khot and Regev(2008)]{khot2008vertex}
S.~Khot and O.~Regev.
\newblock Vertex cover might be hard to approximate to within 2- $\varepsilon$.
\newblock \emph{Journal of Computer and System Sciences}, 74\penalty0
  (3):\penalty0 335--349, 2008.

\bibitem[Kolay et~al.(2019)Kolay, Pandurangan, Panolan, Raman, and
  Tale]{kolay2019harmonious}
S.~Kolay, R.~Pandurangan, F.~Panolan, V.~Raman, and P.~Tale.
\newblock Harmonious coloring: Parameterized algorithms and upper bounds.
\newblock \emph{Theoretical Computer Science}, 772:\penalty0 132--142, 2019.

\bibitem[Lee and Mitchem(2006)]{LeeMitchem}
S.-M. Lee and J.~Mitchem.
\newblock An upper bound for the harmonious chromatic of a graph.
\newblock \emph{Journal of Graph Theory}, 11:\penalty0 565 -- 567, 10 2006.
\newblock \doi{10.1002/jgt.3190110414}.

\bibitem[Lu(1991)]{lu1}
Z.~Lu.
\newblock On an upper bound for the harmonious chromatic number of a graph.
\newblock \emph{Journal of Graph Theory}, 15:\penalty0 345 -- 347, 09 1991.
\newblock \doi{10.1002/jgt.3190150402}.

\bibitem[Mansuri et~al.(2012)Mansuri, Chandel, and Gupta]{article_MYn_CYn}
A.~Mansuri, R.~Chandel, and V.~Gupta.
\newblock On harmonious coloring of {$M(Y_n)$} and {$C(Y_n)$}.
\newblock \emph{World Applied Programming}, 2:\penalty0 150--152, 03 2012.
\newblock ISSN 2222-2510.

\bibitem[Marinescu-Ghemeci(2021)]{link_program}
R.~Marinescu-Ghemeci.
\newblock Exhaustive search program for harmonious coloring.
\newblock \url{https://github.com/veruxy/Harmonious-coloring}, 2021.

\bibitem[McDiarmid and Xinhua(1991)]{McDiarmid1991}
C.~McDiarmid and L.~Xinhua.
\newblock Upper bounds for harmonious colorings.
\newblock \emph{Journal of Graph Theory}, 15:\penalty0 629--636, 1991.

\bibitem[McKay and Royle(1986)]{all_3_regular}
B.~McKay and G.~Royle.
\newblock Constructing the cubic graphs on up to 20 vertices.
\newblock \emph{Ars Combinatoria}, 21a, 01 1986.

\bibitem[Meringer()]{link_33}
M.~Meringer.
\newblock \url{http://www.mathe2.uni-bayreuth.de/markus/reggraphs.html#CRG}.
\newblock February 1997, updated June 2009.

\bibitem[Miller and Pritikin(1991)]{miller1991}
Z.~Miller and D.~Pritikin.
\newblock The harmonious coloring number of a graph.
\newblock \emph{Discrete Mathematics}, 93:\penalty0 211--228, 1991.
\newblock ISSN 0012-365X.
\newblock \doi{https://doi.org/10.1016/0012-365X(91)90257-3}.

\bibitem[Mitchem(1989)]{Mitchem1}
J.~Mitchem.
\newblock On the harmonious chromatic number of a graph.
\newblock \emph{Discrete Mathematics}, 74\penalty0 (1):\penalty0 151 -- 157,
  1989.
\newblock ISSN 0012-365X.
\newblock \doi{https://doi.org/10.1016/0012-365X(89)90207-0}.
\newblock URL
  \url{http://www.sciencedirect.com/science/article/pii/0012365X89902070}.
\newblock Special Double Issue.

\bibitem[Muthumari and Umamamheswari(2016)]{central}
U.~Muthumari and M.~Umamamheswari.
\newblock Harmonious coloring of central graph of some types of graphs.
\newblock \emph{International Journal of Mathematical Archive}, 7\penalty0
  (8):\penalty0 95--103, 2016.
\newblock ISSN 2229 – 5046.

\bibitem[Pratt(1996)]{pratt1996complete}
R.~W. Pratt.
\newblock The complete catalog of 3-regular, diameter-3 planar graphs.
\newblock 1996.

\bibitem[Rajam and Pauline(2013)]{line_star}
K.~Rajam and M.~H.~M. Pauline.
\newblock On harmonious colouring of line graph of star graph families.
\newblock \emph{International Journal of Statistika and Mathematika},
  7:\penalty0 33--36, 2013.
\newblock ISSN 2277- 2790.

\bibitem[Selvi and Azhaguvel(2018)]{jahangir_graph}
M.~Selvi and A.~Azhaguvel.
\newblock A study on harmonious coloring of central graph of jahangir graph.
\newblock \emph{International Journal of Pure and Applied Mathematics},
  118:\penalty0 413--420, 01 2018.

\bibitem[Selvi and Amutha(2020)]{utility}
M.~F.~T. Selvi and A.~Amutha.
\newblock A study on harmonious chromatic number of total graph of central
  graph of generalized petersen graph.
\newblock \emph{Journal of Ambient Intelligence and Humanized Computing}, pages
  1--5, 2020.

\bibitem[Selvi(2015)]{snake_graphs}
M.~S. F.~T. Selvi.
\newblock Harmonious coloring of central graphs of certain snake graphs.
\newblock \emph{Applied Mathematical Sciences}, 9\penalty0 (12):\penalty0
  569--578, 2015.

\bibitem[Takaoka et~al.(2015)Takaoka, Okuma, Tayu, and Ueno]{article}
A.~Takaoka, S.~Okuma, S.~Tayu, and S.~Ueno.
\newblock A note on harmonious coloring of caterpillars.
\newblock \emph{IEICE Transactions on Information and Systems}, E98.D:\penalty0
  2199--2206, 12 2015.
\newblock \doi{10.1587/transinf.2015EDP7113}.

\bibitem[Vazirani(2010)]{Vazirani10}
V.~V. Vazirani.
\newblock \emph{Approximation Algorithms}.
\newblock Springer Publishing Company, Incorporated, 2010.
\newblock ISBN 3642084699.

\bibitem[Vernold et~al.(2012)Vernold, M, and Kaliraj]{double_star}
V.~Vernold, D.~M, and K.~Kaliraj.
\newblock Harmonious coloring on double star graph families.
\newblock \emph{Tamkang Journal of Mathematics}, 43, 06 2012.
\newblock \doi{10.5556/j.tkjm.43.2012.153-158}.

\bibitem[Zhang(2016)]{zhang2016kaleidoscopic}
P.~Zhang.
\newblock \emph{A Kaleidoscopic View of Graph Colorings}.
\newblock SpringerBriefs in Mathematics. Springer International Publishing,
  2016.
\newblock ISBN 9783319305189.
\newblock URL \url{https://books.google.ro/books?id=xlreCwAAQBAJ}.

\end{thebibliography}

\end{document}